\documentclass[ejs]{imsart}

\RequirePackage{amsthm,amsmath,amsfonts,amssymb}
\RequirePackage[numbers]{natbib}
\RequirePackage[colorlinks,citecolor=blue,urlcolor=blue]{hyperref}
\RequirePackage{graphicx}
\usepackage{algorithm}
\usepackage[noend]{algpseudocode}

\arxiv{2010.00000}
\startlocaldefs
\theoremstyle{plain}

\newtheorem{theorem}{Theorem}[section]
\newtheorem{lemma}[theorem]{Lemma}
\newtheorem{proposition}{Proposition}[section]
\theoremstyle{definition}

\newtheorem{assumption}{Assumption}
\theoremstyle{remark}

\newcommand{\boundmeas}{\mathcal{B}}
\newcommand{\ESS}{\mathcal{E}}

\newcommand{\E}{\mathbb{E}}

\newcommand{\N}{\mathbb{N}}
\newcommand{\X}{\mathbb{X}}
\newcommand{\Y}{\mathbb{Y}}
\newcommand{\bi}{\boldsymbol{i}}
\newcommand{\bj}{\boldsymbol{j}}
\newcommand{\bu}{\boldsymbol{u}}
\newcommand{\bv}{\boldsymbol{v}}
\newcommand{\bvarphi}{\boldsymbol{\varphi}}

\newcommand{\bH}{\boldsymbol{H}}
\newcommand{\bK}{\boldsymbol{K}}
\newcommand{\bX}{\mathcal{X}}
\newcommand{\bY}{\mathcal{Y}}

\renewcommand{\P}{\mathbb{P}}
\newcommand{\R}{\mathbb{R}}
\newcommand{\bR}{\boldsymbol{R}}
\newcommand{\cR}{\mathcal{R}}

\newcommand{\ud}{\mathrm{d}}
\newcommand{\ug}{\underline{g}}
\newcommand{\uK}{\underline{K}}
\newcommand{\un}{\underline{n}}
\newcommand{\upi}{\underline{\pi}}
\newcommand{\ugamma}{\underline{\gamma}}
\newcommand{\uw}{\underline{w}}
\newcommand{\uW}{\underline{W}}
\newcommand{\uZ}{\underline{Z}}
\newcommand{\uQ}{\underline{Q}}

\newcommand{\uxi}{\underline{\xi}}
\newcommand{\barphi}{\overline{\varphi}}

\newcommand{\ol}[1]{\overline{#1}}

\newcommand{\ones}{\boldsymbol{1}}
\newcommand{\eye}{\mathbf{I}}
\newcommand{\bxi}{\boldsymbol{\xi}}
\newcommand{\osc}{\mathrm{osc}}

\newcommand{\Var}{\mathrm{Var}}

\endlocaldefs

\begin{document}
\begin{frontmatter}
\title{Augmented Island Resampling Particle Filters for Particle Markov Chain Monte Carlo}
\runtitle{AIRPF for PMCMC}

\begin{aug}
%%%%%%%%%%%%%%%%%%%%%%%%%%%%%%%%%%%%%%%%%%%%%%%
%% Only one address is permitted per author. %%
%% Only division, organization and e-mail is %%
%% included in the address.                  %%
%% Additional information can be included in %%
\author[A]{\fnms{Kari}~\snm{Heine}\ead[label=e1]{kmph20@bath.ac.uk}\orcid{0000-0001-6300-3258}}
%%%%%%%%%%%%%%%%%%%%%%%%%%%%%%%%%%%%%%%%%%%%%%
%% Addresses                                %%
%%%%%%%%%%%%%%%%%%%%%%%%%%%%%%%%%%%%%%%%%%%%%%
\address[A]{Department of Mathematical Sciences,
University of Bath\printead[presep={,\ }]{e1}}

\runauthor{K. Heine}
\end{aug}

\begin{abstract}
In modern days, the ability to carry out computations in parallel is key to efficient implementations of computationally intensive algorithms. This paper investigates the applicability of the previously proposed Augmented Island Resampling Particle Filter (AIRPF) --- an algorithm designed for parallel implementations --- to particle Markov Chain Monte Carlo (PMCMC). We show that AIRPF produces a non-negative unbiased estimator of the marginal likelihood and hence is suitable for PMCMC. We also prove stability properties, similar to those of the $\alpha$SMC algorithm, for AIRPF. This implies that the error of AIRPF can be bound uniformly in time by controlling the effective number of filters, which in turn can be done by adaptively constraining the interactions between filters. We demonstrate the superiority of AIRPF over independent Bootstrap Particle Filters, not only numerically, but also theoretically. To this end, we extend the previously proposed collision analysis approach to derive an explicit expression for the variance of the marginal likelihood estimate. This expression admits exact evaluation of the variance in some simple scenarios as we shall also demonstrate.
\end{abstract}



\end{frontmatter}
%%%%%%%%%%%%%%%%%%%%%%%%%%%%%%%%%%%%%%%%%%%%%%
%% Please use \tableofcontents for articles %%
%% with 50 pages and more                   %%
%%%%%%%%%%%%%%%%%%%%%%%%%%%%%%%%%%%%%%%%%%%%%%
%\tableofcontents

\section{Introduction}
\label{sec:intro}

Particle Markov chain Monte Carlo (PMCMC)~\citep{andrieu_et_al10} is a computational methodology with a wide range of applications, see e.g.~\cite{kokkala_sarkka15, karppinen19, nguyen23}, for joint Bayesian estimation of a latent signal process together with the model parameters. In a nutshell, PMCMC refers to an MCMC algorithm that deploys a sequential Monte Carlo (SMC) algorithm~\cite{doucet_et_al01} within each MCMC iteration to produce a non-negative and unbiased estimate of the marginal likelihood (marginalised over the latent signal process) for a given parameter value. Reminiscent to the pseudo-marginal MCMC of \cite{andrieu_roberts09}, the unbiased marginal likelihood estimate is used to compute the correct Metropolis-Hastings acceptance probability which ensures that the resulting Markov chain targets the joint Bayesian posterior distribution for both the latent signal process and the model parameters. 

SMC algorithms are notorious for their computational cost. In the context of PMCMC, they may be run thousands of times and thus, the computational efficiency of the deployed SMC algorithm is of utmost importance. To make algorithms computationally efficient by leveraging the power of modern computing systems, they must be suitable for parallel processing. The first SMC algorithms, such as the classical Bootstrap Particle Filter (BPF)~\cite{gordon_et_al93}, are inherently sequential, but various suggestions to parallelise particle filters have been made throughout the past decades \cite{verge_et_al15,heine_et_al20,bolic_et_al05,whiteley_et_al16}. In this paper, we focus specifically on the \emph{Augmented Island Resampling Particle Filter (AIRPF)} of \cite{heine_et_al20} --- an SMC algorithm specifically designed for parallel computing systems --- and its application to PMCMC. 

\subsection{Contributions}

AIRPF is known to be consistent and stable~\cite{heine_et_al20}, but in this paper we shall expand the analysis of AIRPF from the PMCMC point of view, first by carrying out the relatively straightforward analysis to show that AIRPF also produces an unbiased estimator of the Feynman-Kac normalisation coefficient, and hence is suitable for PMCMC. 

A generic SMC algorithm, $\alpha$SMC, was introduced in \cite{whiteley_et_al16} and it admits e.g.~BPF, sequential importance sampling (SIS), and the adaptive resampling BPF algorithms as specific instances. The consistency and stability analysis carried out in \cite{whiteley_et_al16} for $\alpha$SMC established a connection between the uniform convergence of $\alpha$SMC and the \emph{effective sample size (ESS)}~\cite{kong_et_al94}; by controlling the ESS to keep it above a given threshold, the approximation error can be bound uniformly in time. This analysis not only gave an immediate stability result of for adaptive resampling schemes due to the control on the ESS, but it also enabled a rigorous argument in favour of interacting parallel particle filters over independet BPFs \cite[see Section 5]{whiteley_et_al16}. Our second contribution is to extend the stability result of \cite{whiteley_et_al16} to AIRPF, which is not an $\alpha$SMC algorithm. This analysis forms the core of this paper and it enables a simple and rigorous line of arguments to suggest the superiority of AIRPF over multiple independent BPFs.

The third contribution if this paper is a rigorous derivation of the formulae for the marginal likelihood estimate variance for AIRPF and comparison with IBPF and the augmented resampling particle filter (ARPF) of \cite{heine_et_al20}. This derivation follows the collision analysis of \cite{cerou_et_al11} which was further developed in \cite{heine_et_al17} and although the resulting formulae can only be evaluated in some academic toy applications, they shed light on the theoretical mechanisms governing the impact of the filter interactions on the error of the PMCMC estimates.

\subsection{Formal problem statement}
\label{sec:problem statement}

Let $X^{\theta} = (X^{\theta}_{0},X^{\theta}_{1},\ldots)$ be a latent discrete time Markovian signal process, and let $Y^\theta = (Y^{\theta}_{0},Y^{\theta}_{1},\ldots)$ be an observation process. $X^{\theta}$ and $Y^\theta$ are assumed to take values in sufficiently regular measurable spaces $(\X,\bX)$ and $(\Y,\bY)$, respectively, with laws specified by
\begin{align}
\begin{array}{rll}
    X^{\theta}_{n} \mid X^{\theta}_{n-1} = x &\sim K_n^\theta(x,\,\cdot\,), \quad X^{\theta}_{0} \sim \pi^\theta_{0}, & \qquad \forall~n\in \N := \{1,2,\ldots\}, \\[2mm]
        Y^{\theta}_{n} \mid X^{\theta}_{n} = x &\sim G^{\theta}_n(x,\,\cdot,),& \qquad \forall~n \in \N_0 := \{0,1,\ldots\},
\end{array}\label{eq:hmm}
\end{align}
where $\pi^{\theta}_{0}$ is a probability measure on $\bX$, and  $K^\theta_n:\X\times\bX\to[0,1]$ and $G_n^\theta:\X,\bY\to[0,1]$ are probability kernels. The kernels and the initial distribution --- and consequently $X^{\theta}$ and $Y^\theta$ --- are assumed to be parameterised by $\theta \in \Theta \subset \mathbb{R}^d$. We fix $Y^\theta$ to a given realisation $(y_0,y_1,\ldots)$, and write $g^{\theta}_{n} = G_n^{\theta}(\,\cdot\,,y_n)$ where $G_n^{\theta}$ denotes a conditional probability density w.r.t.~some $\sigma$-finite measure on $\bY$, typically the Lebesgue measure.   

The model \eqref{eq:hmm} gives rise to the Feynman-Kac measures (see e.g.~\cite{delmoral04})
\begin{align}\nonumber
    \pi_{n}^{\theta} = \gamma_{n}^{\theta} / \gamma_{n}^{\theta}(1), \qquad \forall\,n\in \N_0,
    % \pi_{n}^[\theta}(\varphi) = \frac{\gamma_{n}^{\theta}(\varphi)}{\gamma_{n}^{\theta}(1)}, \qquad \forall\,n\in \N_0,
\end{align}
where the unnormalised Feynman-Kac measures $\gamma_{n}^{\theta}$ are defined as 
\begin{align}\label{eq:unnormalised FK measure}
    \gamma_{n}^{\theta}(\varphi) = \E_{\theta}\!\left[\varphi(X^{\theta}_{n})\prod_{q=0}^{n-1} g^{\theta}_{q}(X^{\theta}_q)\right], \qquad \forall~\varphi \in \boundmeas(\X),~n\in \N_0,
\end{align}
where $\boundmeas(\X)$ denotes the set of bounded and measurable real valued functions on $\X$, and $\E_\theta[\,\cdot\,]$ denotes the expectation over the law of $X^\theta$ for given $\theta$. The measures 
\begin{align*}
\pi^{\theta}_{n}(\ud x) = \P(X^{\theta}_n \in \ud x \mid Y^{\theta}_0=y_0,\ldots, Y^{\theta}_{n-1} = y_{n-1}),
\end{align*}
are known as the prediction filter and they are typically the object of interest in Bayesian inference on hidden Markov models.

The Particle Marginal Metropolis-Hastings (PMMH) algorithm \cite{andrieu_et_al10} approximates the joint Bayesian posterior distribution of $\theta$ and $X^\theta$ by deploying SMC to construct sample based approximations of $\pi_{n}^{\theta}$ and the marginal likelihood $\gamma_{n}^{\theta}(1)$. In this paper, we shall analyse the approximations of $\pi_{n}^{\theta}$ and $\gamma_{n}^{\theta}(1)$ obtained with AIRPF and show that not only is AIRPF suitable for PMMH, but in certain situations preferable over independent BPFs (IBPF). From now on --- as we shall focus on the SMC part of PMMH --- we shall simplify the notation by suppressing the explicit dependency on $\theta$ in all notations.

\subsection{Organisation}

Section \ref{sec:methodology} describes the AIRPF algorithm and contains the theoretical analysis to establish the lack-of-bias, consistency and stability. Section \ref{sec:AIRPF vs IBPF} carries out a comparison between IBPF and AIRPF, including the derivation of the collision analysis based variance formula for the marginal likelihood estimate of AIRPF. Section \ref{sec:numerical experiments} concludes the paper with some illustrative numerical experiments.

% -------------------------------------------
% -------------------------------------------
%   AIRPF
% -------------------------------------------
% -------------------------------------------
\section{Augmented Island Resampling Particle Filter}
\label{sec:methodology}

AIRPF runs $m \in \N$ interacting particle filters in parallel, the case $m=1$ being equivalent to a single BPF. The filters deploy a sample of size $M \in \N$ each, making the total sample size $N=mM$. In general, $m$ is an arbitrary integer power of the parameter $r\in \N \setminus \{1\}$, but as in \cite{heine_et_al20}, we shall focus on the case $r = 2$ for simplicity. The analysis does not rely on this assumption.

The $m$ filters interact in $S_m = \log_{2}(m)$ stages, according to  the matrices %$A_1,\ldots,A_{S_m}$ defined as
\begin{align}\label{eq:alpha matrices}
    A_s = \eye_{2^{S_m-s}} \otimes \frac{1}{2}\ones_{2} \otimes \eye_{2^{s-1}}%\otimes \frac{1}{M}\ones_{M}
    ,\quad s \in \{1,\ldots,S_m\},
\end{align}
where $\otimes$ denotes the Kronecker product and, for all $\ell \in \N$, the matrices $\eye_{\ell}$ and $\ones_{\ell}$ denote the size $\ell$ identity matrix, and the size $\ell$ matrix of ones, respectively.  AIRPF algorithm is given in Algorithm \ref{alg:AIRPF} below, where we also use the  notations $\bxi^{k}_{n,s} := (\xi^{k,1}_{n,s},\ldots,\xi^{k,M}_{n,s})$ %, $\bXi_{n,s} := (\bxi^{1}_{n,s},\ldots,\bxi^{m}_{n,s})$ 
and
\begin{align*}
\bK_n(\bxi^{k}_{n,s},\,\cdot\,) := K_n^{\otimes M}(\bxi^{k}_{n,s},\,\cdot\,) := K_n(\xi^{k,1}_{n,s},\,\cdot\,)\otimes \cdots \otimes K_n(\xi^{k,M}_{n,s},\,\cdot\,),
\end{align*}
where $\otimes$ denotes a measure product. We also let $\ones(\,\cdot\,)$ denote the indicator function.
\begin{algorithm}[h!]
\caption{AIRPF}
\label{alg:AIRPF}
\begin{algorithmic}[1]
\For{$k \in \{1,\ldots,m\}$} \quad \verb+%% Initialisation+
\State $W^{k}_{0,0} = 1$,\quad $\bxi^{k}_{0,0} \sim \pi_0^{\otimes M} := \pi_0 \otimes \cdots \otimes \pi_0$
\EndFor
\State
\For{$n \in \{0,1,\ldots\}$} \quad \verb+%% Filter main loop+
\State $\ESS_{n,0}^{mM} = {\left(m^{-1}\sum_{k=1}^{m}W_{n,0}^k\right)^2}/{m^{-1}\sum_{k=1}^{m}{(W_{n,0}^k)}^2}$ \quad \verb+%% ENF+
\State
\For{$k \in \{1,\ldots,m\}$} \quad \verb+%% Iterate over filters+
\State $W_{n,1}^k \leftarrow {M^{-1}}\sum_{i=1}^MW^{k}_{n,0}g_n(\xi^{k,i}_{n,0})$ \label{code line 3} \quad \verb+%% Marginal likelihood estimation+
\State
\For{$i \in \{1,\ldots,M\}$}\label{code line 1} \quad
 \verb+%% Internal resample+
\State $\xi^{k,i}_{n,1} \sim {M^{-1}}\sum_{i=1}^{M}W^{k}_{n,0}g_n(\xi^{k,i}_{n,0})\delta_{\xi^{k,i}_{n,0}}/W_{n,1}^k$\label{code line 2}
\EndFor
\EndFor
\State
\For{$s \in \{1,\ldots,S_m\}$} \quad \verb+%% Filter interaction: iterate over stages+
\State $\ESS_{n,s}^{mM} = {\left(m^{-1}\sum_{k=1}^{m}W_{n,s}^k\right)^2}/{m^{-1}\sum_{k=1}^{m}{(W_{n,s}^k)}^2}$ \quad \verb+%% ENF+
\State $\alpha_s = A_s \ones(\ESS_{n,s}^{mM} < \tau) + \eye_{m}\ones(\ESS_{n,s}^{mM} \geq \tau)$ \label{adaptation line}
\State
\For{$k \in \{1,\ldots,m\}$} \quad \verb+%% Iterate over filters+
\State $W^{k}_{n,s+1} \leftarrow \sum_{\ell=1}^{m} \alpha^{k\ell}_{s}W^{\ell}_{n,s}$ \label{line:normaliser}
\If{$s <  S_m$}
    \State $\bxi_{n,s+1}^{k} \sim (W^{k}_{n,s+1})^{-1}\sum_{\ell=1}^{m} \alpha^{k\ell}_{s}W^{\ell}_{n,s}\delta_{\bxi^{\ell}_{n,s}}$\label{code line 4}
\Else \quad \verb+%% Include mutation to the final stage+ 
    \State $W^{k}_{n+1,0} = W^{k}_{n,s+1}$,\quad$\bxi_{n+1,0}^{k} \sim (W^{k}_{n,s+1})^{-1}\sum_{\ell=1}^{m} \alpha^{k\ell}_{s}W^{\ell}_{n,s}\bK_{n+1}(\bxi^{\ell}_{n,s},\,\cdot\,)$\label{code line 5}%\todo{}
\EndIf
\EndFor
\EndFor
\EndFor
\end{algorithmic}
\end{algorithm} 

Our focus is on the interactions between filters rather than the interactions between individual particles, and therefore each filter is assumed to do internal resampling at each iteration. However, we do allow adaptive interactions between filters, as shown on line \ref{adaptation line} of Algorithm \ref{alg:AIRPF} which is a slight deviation from the definition of AIRPF in \cite{heine_et_al20}. The filters interact, if the \emph{effective number of filters (ENF)}, denoted by $\ESS_{n,s}^{mM}$, goes below a threshold value $\tau \in (0,1]$. ENF is analogous to ESS, but based on filter weights rather than particle weights, giving it an intuitive interpretation of representing the number of filters with non-negligible weights. Our analysis does not rely on internal resampling being done at each iteration, and thus the results should hold also for adaptive internal resampling schemes.

AIRPF produces weighted empirical measures
\begin{align}\label{eq:output measure}
    \pi^N_{n,s} = \frac{\sum_{k=1}^m W^{k}_{n,s}\frac{1}{M}\sum_{i=1}^{M}\delta_{\xi^{k,i}_{n,s}}}{\sum_{k=1}^{m}W^{k}_{n,s}}, \qquad \forall~n \in \N_0,~s \in \{0,\ldots,S_m\},
\end{align}
from which an approximation of $\pi_n$ is obtained as $\pi^{N}_{n} = \pi^{N}_{n,0}$. The unnormalised Feynman-Kac measure approximations, that yield the marginal likelihood estimates, are
\begin{align*}
\gamma_{n}^{N}(\varphi) = \frac{1}{m}\sum_{k=1}^m W^{k}_{n,0}\frac{1}{M}\sum_{i=1}^{M}\delta_{\xi^{k,i}_{n,0}} \qquad \forall~n \in \N_0.%\pi^{N}_n(\varphi)\prod_{p=0}^{n-1}\pi^{N}_{p}(g_p), \qquad \forall~n \in \N_0.
\end{align*}

% ------------------------------------------------
% LACK-OF-BIAS
% ------------------------------------------------
\subsection{Lack-of-bias}
\label{sec:lack-of-bias}

Define kernels $Q_n, Q_{p,n}: (\X,\bX) \to [0,+\infty)$ for all $n\in \N$ and $0 \leq p < n$ as
\begin{align*}
    Q_{n}(x,\ud x') = g_{n-1}(x)K_n(x,\ud x'), \qquad \text{and} \qquad Q_{p,n} = Q_{p+1}\cdots Q_{n}, 
\end{align*}
with the convention that $Q_{n,n} = \mathrm{Id}$, i.e.~the identity kernel, for any $n\in \N$. In this case, 
\begin{align*}
    \gamma_{n} = \gamma_pQ_{p,n} = \gamma_{n-1}Q_{n}, \qquad 0 \leq p \leq n \in \N.
\end{align*}

\begin{theorem}\label{thm:lack-of-bias}
% \label{prop:unbiasedness of island augmented resample}
Fix $m,~M\in \N$ and let $N=mM$ and $\pi^{N}_{n} = \pi^{N}_{n,0}$ be as defined in \eqref{eq:output measure}. Then, 
\begin{align*}%\label{eq:interm unbiasedness}
    \E\left[ {\pi}^N_{n+1}(\varphi) \,\middle|\, \mathcal{F}_{n}\right] 
    = \frac{\pi^{N}_{n}(g_{n}K_{n+1}(\varphi))}{\pi^{N}_{n}(g_{n})}\qquad \forall~n \in \N_0,~\varphi \in \boundmeas(\X),
\end{align*}
where $\mathcal{F}_{n} := \sigma(\bxi_{0,0},\ldots,\bxi_{n,0})$ and $\bxi_{p,s} := (\bxi^{1}_{p,s},\ldots,\bxi^{m}_{p,s})$ for all $s \in \{0,\ldots,S_m\}$ and $p\in\N_0$.
\end{theorem}
\begin{proof}
By lines \ref{code line 3} and \ref{code line 2} in Algorithm \ref{alg:AIRPF} and the fact that $W^{k}_{n,1}$ is $\mathcal{F}_{n}$-measurable, we have
\begin{align}
    \E\left[{\pi}^{N}_{n,1}(\varphi) \,\middle|\, \mathcal{F}_{n}\right] 
    &=
    \frac{\sum_{k=1}^{m}W^{k}_{n,1}\frac{1}{M}\sum_{i=1}^{M}\E\left[\varphi\big(\xi^{k,i}_{n,1}\big)\,\middle|\, \mathcal{F}_n\right]}{\sum_{k=1}^{m}W^{k}_{n,1}} \nonumber \\
    &=
    \frac{\sum_{k=1}^{m}W^{k}_{n,0}\frac{1}{M}\sum_{i=1}^{M}g_{n}(\xi^{k,i}_{n,0})\varphi(\xi^{k,i}_{n,0})}{\sum_{k=1}^{m}W^{k}_{n,0}\frac{1}{M}\sum_{i=1}^{M}g_{n}(\xi^{k,i}_{n,0})} \nonumber\\ &= \frac{\pi^{N}_{n,0}(g_n\varphi)}{\pi^{N}_{n,0}(g_n)}. \label{eq:unbiasedness stage 0}
\end{align}
Let us write $\mathcal{F}_{n,s}= \mathcal{F}_n \vee \sigma(\bxi_{n,1},\ldots,\bxi_{n,s})$ for all $s \in \{1,\ldots,S_m\}$. Similarly to \eqref{eq:unbiasedness stage 0}, by the lines \ref{line:normaliser}, \ref{code line 4} and \ref{code line 5} of Algorithm \ref{alg:AIRPF}, the $\mathcal{F}_{n,s-1}$-measurability of $W_{n,s}^{k}$, for all $s \in \{2,\ldots,S_m\}$, and the fact that $\sum_{k=1}^m\alpha_{s}^{k,\ell}=1$ for any $\ell\in \{1,\ldots,m\}$ and $s \in \{1,\ldots,S_m\}$, we have
\begin{align}
    \E\left[{\pi}^{N}_{n,s}(\varphi)\,\middle|\, \mathcal{F}_{n,s-1}\right]  
    =
    {\pi}^{N}_{n,s-1}(\varphi) \label{eq:unbiasedness of stage s},
\end{align}
Finally, by \eqref{eq:output measure} and by observing that 
\begin{align*}
    \E\left[\varphi(\xi^{k,i}_{n+1,0})\,\middle|\, \mathcal{F}_{n,S_m}\right] = \frac{\sum_{\ell=1}^{m}\alpha^{k,\ell}_{S_m}W^{\ell}_{n,S_m}K_{n+1}(\varphi)(\xi^{\ell,i}_{n,S_m})}{\sum_{\ell=1}^m\alpha^{k,\ell}_{S_m}W^{\ell}_{n,S_m}},
\end{align*}
and that $W^{k}_{n+1,0} = \sum_{\ell=1}^m\alpha^{k,\ell}_{S_m}W^{\ell}_{n,S_m}$, we obtain, similarly to above,  
\begin{align}\label{eq:unbiasedness stage S+1}
    \E\left[ {\pi}^N_{n+1,0}(\varphi) \mid \mathcal{F}_{n,S_m}\right] = {\pi}^N_{n,S_m}(K_{n+1}(\varphi)).
\end{align}
The claim follows from \eqref{eq:unbiasedness stage 0} -- \eqref{eq:unbiasedness stage S+1} and the tower property of conditional expectations.
\end{proof}

By setting $K_{n+1}=\mathrm{Id}$ and $\varphi(x) = \ones[\xi^{k,i}_{n,0}=x]$ for any $k \in \{1,\ldots,m\}$ and $i \in \{1,\ldots,M\}$, Theorem \ref{thm:lack-of-bias} implies the lack-of-bias condition of~\cite[Assumption 2]{andrieu_et_al10} for resampling. Thus the premises set in~\cite[Assumption 2]{andrieu_et_al10} are met, and we can conclude AIRPF to be suitable for PMMH according to \cite[Theorem 4]{andrieu_et_al10}. Moreover, by~\cite[Proposition 7.4.1]{delmoral04}, Theorem \ref{thm:lack-of-bias} also yields the lack-of-bias property $\E[\gamma^N_{n}(1)] = \gamma_{n}(1).$

\subsection{Augmented Feynman-Kac model}

Our extension of Theorem 2 of \cite{whiteley_et_al16} to AIRPF is based on interpreting the measures $\pi^{N}_{0}\ldots,\pi^{N}_{n}$, obtained with Algorithm~\ref{alg:AIRPF}, as a subsequence of the normalised measures $\underline{\pi}^N_0,\ldots,\underline{\pi}^N_{(S_m+1)n}$ of \emph{an augmented Feynman-Kac model}, obtained from the original model by regarding the $S_m$ stages of AIRPF as an artificial Feynman-Kac updates involving weighting with a constant potential function, signifying a non-informative observation, and a mutation with an identity kernel, to signify that the signal cannot mutate over the augmented stages. 

Formally, the potential functions $\ug_{\rho}:\X\to[0,+\infty)$ and kernels $\uK_{\rho}:\X \times \bX \to [0,1]$ of the augmented Feynman-Kac model are defined for all $\rho \in \{0,\ldots,(S_m+1)n\}$ as
\begin{align}\label{eq:def aug potential and kernel}
    \ug_{\rho} = g_{I_{S_m+1}(\rho),R_{S_m+1}(\rho)}\qquad \text{and}\qquad \uK_{\rho} = K_{I_{S_m+1}(\rho),R_{S_m+1}(\rho)},
\end{align}
where $I_q(p) = \left\lfloor p/q \right\rfloor$ and $R_q(p) = p - qI_q(p)$, 
and for all $s \in \{0,\ldots,S_m\}$ and $p \in \N_0$
\begin{align}\label{eq:def 2 index aug potential and kernel}
g_{p,s} := g_p\ones(s=0)+\ones(s>0)\quad \text{and}\quad K_{p,s} := K_p\ones(s=0) + \mathrm{Id}\ones(s>0).
%\arraycolsep=1.4pt
%\begin{array}{rl}
%g_{p,s} &:= g_p\ones(s=0)+\ones(s>0),\\ K_{p,s}(x,\ud x') &:= K_p(x,\ud x')\ones(s=0) + \delta_{x}(\ud x')\ones(s>0).
%\end{array}
\end{align}
Analogously to Section \ref{sec:lack-of-bias}, we define $\uQ_\rho$ and $\uQ_{\rho,n}$ for all $n\in \N_0,~\rho \in \{0,\ldots,n\}$ in terms of $\ug_{\rho}$ and $\uK_{\rho}$. With $\underline{\pi}_0 = \pi_0$, the Feynman-Kac measures for the augmented model are
\begin{align*}
    \underline{\gamma}_n = \underline{\pi}_0\underline{Q}_{0,n}\quad\text{and}\quad \underline{\pi}_{n} = {\underline{\gamma}_{n}}/{\underline{\gamma}_{n}(1)}\qquad \forall~n\in \N_0,
\end{align*}
and thus $\pi_{p} = \underline{\pi}_{(S_m+1)p}$ and $\gamma_{p} = \underline{\gamma}_{(S_m+1)p}$ for all $p \in \N_0$.

In order to accommodate the augmented Feynman-Kac model interpretation for the particle approximations, the particles in Algorithm \ref{alg:AIRPF} will be indexed as
\begin{align}\label{eq:aug particle and weight indexing}
    \uxi^{(k-1)M+i}_{\rho} = \uxi^{k,i}_{\rho} = \xi^{k,i}_{I_{S_m+1}(\rho),R_{S_m+1}(\rho)} \quad \text{and} \quad \uW^{k}_{\rho} = W^{k}_{I_{S_m+1}(\rho),R_{S_m+1}(\rho)},
\end{align}
for all $\rho \in \{0,1,\ldots,(S_m+1)n\}$, $i \in \{1,\ldots,M\}$ and $k \in \{1,\ldots,m\}$. Although the particle weights within filter are constant, it is useful to define individual particle weights as 
\begin{align*}
w^{(k-1)M+i}_{p,s} = w^{k,i}_{p,s} = W^{k}_{p,s}\qquad\text{and}\qquad \uw^{(k-1)M+i}_{\rho} = \uw^{k,i}_{\rho} = \uW^{k}_{\rho}
\end{align*}
for the original and the augmented model, respectively.  Finally, analogously to the definitions in \cite{whiteley_et_al16}, we define scaled weights
\begin{align*}
    \ol{w}^{(k-1)M+i}_{p,s} = \ol{w}^{k,i}_{p,s} = w^{k,i}_{p,s}/\gamma_{p}(1)\qquad\text{and}\qquad\ol{\uw}^{(k-1)M+i}_{p} = \ol{\uw}^{k,i}_p = \uw^{k,i}_{p}/\ugamma_{p}(1)
\end{align*}
for all $k \in \{1,\ldots,m\}$, $i \in \{1,\ldots,M\}$, and $p \in \N$, and the scaled kernels
\begin{align*}
    \ol{Q}_{p,n} = \frac{Q_{p,n}}{\pi_p Q_{p,n}(1)} \quad \text{and}\quad\ol{\uQ}_{p,n} = \frac{\uQ_{p,n}}{\upi_p \uQ_{p,n}(1)},\qquad \forall~n \in \N_0,~p \in \{0,\ldots,n\}.
\end{align*}

\subsection{Consistency and Stability}
\label{sec:stability}

Our analysis is carried out under a standard regularity assumption, similarly to  \cite{whiteley_et_al16}.
% --------------------------------------
\begin{assumption}\label{ass:regularity}
There exists $\delta,\epsilon\in [1,+\infty)$ such that,
\begin{align*}
    \sup_{n\geq 0 }\sup_{x,y} \frac{g_n(x)}{g_n(y)} \leq \delta\quad\text{and}\quad K_{n}(x,\,\cdot\,) \leq \epsilon K_{n}(y,\,\cdot\,),\quad\forall (x,y) \in \X^2,~n\in \N.
\end{align*}
\end{assumption}

% -----------------------------------------------------
% MAIN MARTINGALE RESULT
% -----------------------------------------------------
\begin{proposition}\label{prop:martingale construction}
For any $m, M \in \N$, $n \in \N_0$, $\varphi \in \boundmeas(\X)$, and $N=mM$, we have the decomposition
\begin{align}\label{eq:martingale difference decompo}
    \mathcal{D}(n) = \frac{1}{N}\sum_{i=1}^N \ol{w}^{i}_{n,0}\varphi(\xi^{i}_{n,0}) - \pi_{n}(\varphi) = \sum_{\rho=0}^{\underline{n}} D_{\rho,\un},
\end{align}
where $\un = (S_m+1)n$, and for all $0 \leq \rho \leq q \in \N_0$,
\begin{align*}
D_{\rho,q} = 
\displaystyle\frac{1}{N}\sum_{i=1}^{N}\ol{\uw}_{\rho}^{i}\ol{\uQ}_{\rho,q}(\varphi)({\uxi}^{i}_{\rho}) - \frac{1}{N}\sum_{i=1}^{N}\ol{\uw}_{\rho-1}^{i}\ol{\uQ}_{\rho-1,q}(\varphi)({\uxi}^{i}_{\rho-1}), %\qquad \forall~ 0 \leq \rho \leq q \in \N_0,
\end{align*}
with the convention that $\ol{\uw}_{-1}^{i}\ol{\uQ}_{-1,q}(\varphi)(\uxi^{i}_{-1}) = \upi_{q}(\varphi)$. 
In addition, by defining $\underline{\mathcal{G}}_{-1} = \{\X,\emptyset\}$ and $\underline{\mathcal{G}}_{\rho} := \underline{\mathcal{G}}_{\rho-1}\vee\sigma(\uxi^{1}_{\rho},\ldots,\uxi^{N}_{\rho})$ for all $0 \leq \rho \leq q \in \N_{0}$, we have $\E\big[D_{\rho,q} \,|\, \underline{\mathcal{G}}_{\rho-1}\big] = 0$ almost surely,   
%\begin{align}\label{eq:martingale difference property}
%\end{align} 
and
\begin{multline}\label{eq:martingale error bound}
    \E\left[\left|\mathcal{D}(n)\right|^r\right]^{\frac{1}{r}} 
    \leq 
    \frac{C_{r}}{\sqrt{N}}\osc\big(\ol{\uQ}_{0,\un}(\varphi)\big) + \\ \frac{C_r}{N}
    \sum_{p=0}^{n-1} \E\!\left[\!\left(\sum_{s=1}^{S_m+1}\osc\big(\ol{\uQ}_{p,s,\un}(\varphi)\big)^2\sum_{i=1}^{N}\left(\ol{w}_{p,s}^{i}\right)^2\right)^{\frac{r}{2}}\right]^{\frac{1}{r}}
\end{multline}
where $S_m = \log_2(m)$ and $\overline{\uQ}_{p,s,\un} = \overline{\uQ}_{(S_m+1)p+s,\un}$ for all $p \in \{0,\ldots,n\},~s \in \{0,\ldots,S_{m}\}$.
\end{proposition}
\begin{proof}
Equation  \eqref{eq:martingale difference decompo} is trivial. Define $\cR_{0}^{q} = \{\rho \in \{0,\ldots,q\}: R_{S_m+1}(\rho-1) = 0\}$, $\cR_{+}^{q} = \{\rho \in  \{0,\ldots,q\}: R_{S_m+1}(\rho-1) > 0\}$ and for all $i \in \{1,\ldots,M\}$ and $k \in \{1,\ldots,m\}$,
\begin{multline}\label{eq:2nd delta}
\Delta^{k,i}_{\rho,q}(\varphi)  \\ =
\begin{cases}
\displaystyle\ol{\uQ}_{0,q}(\varphi)(\uxi^{k,i}_{0}) - \frac{1}{\ol{\uw}_{0}^{k,i}}\upi_q(\varphi), &\rho = 0\\
\displaystyle\ol{\uQ}_{\rho,q}(\varphi)(\uxi^{k,i}_{\rho}) - \frac{1}{\ol{\uw}_{\rho}^{k,i}}\sum_{j=1}^{N}\alpha_{0}^{(k-1)M+i,j}\ol{\uw}_{\rho-1}^{j}\ol{\uQ}_{\rho-1,q}(\varphi)(\uxi^{j}_{\rho-1}), & \rho \in \cR_0^{q},\\
\displaystyle\ol{\uQ}_{\rho,q}(\varphi)(\uxi^{k,i}_{\rho}) - \frac{1}{\ol{\uw}_{\rho}^{k,i}}\sum_{\ell=1}^{m}\alpha_{R_{S_m+1}(\rho-1)}^{k,\ell}\ol{\uw}_{\rho-1}^{\ell,i}\ol{\uQ}_{\rho-1,q}(\varphi)(\uxi^{\ell,i}_{\rho-1}),& \rho \in \cR_+^{q},
\end{cases}
\end{multline}
where $\alpha_0 = \eye_{m} \otimes {M}^{-1}\ones_{M}$. With these notations we have
\begin{align*}
D_{\rho,q} = \frac{1}{N}\sum_{k=1}^{m}\sum_{i=1}^{M} \ol{\uw}^{k,i}_{\rho}\Delta^{k,i}_{\rho,q}(\varphi).
\end{align*}
Let us first show that $\E[D_{\rho,q} \mid \underline{\mathcal{G}}_{\rho-1}] = 0$, almost surely, for all $ 0 \leq \rho \leq q \in \N_0$.

For $\rho = 0$, since $\ol{\uw}_{0}^{i} = 1$ for all $i \in \{1,\ldots,N\}$ and $(\uxi^{1}_{0},\ldots,\uxi^{N}_{0})\sim \pi_0^{\otimes N}$, we have
\begin{align*}
\E\!\left[\overline{\underline{w}}_{0}^{i}\ol{\uQ}_{0,q}(\varphi)(\underline{\xi}^{i}_{0})\right] = \pi_{0}\ol{\uQ}_{0,q}(\varphi) = \frac{\upi_{0}{\uQ}_{0,q}(\varphi)}{\upi_{0}\uQ_{0,q}(1)} = \upi_{q}(\varphi),
\end{align*}
and hence $\E[D_{0,q} \mid \underline{\mathcal{G}}_{-1}] = 0$.  

Consider $\rho \in \cR_0^q$ and write $p = I_{S_m+1}(\rho-1)$.  Since $R_{S_m+1}(\rho-1) = 0$, by \eqref{eq:def aug potential and kernel} -- \eqref{eq:aug particle and weight indexing} we have\footnote{We can assume that $m > 1$, implying $S_m + 1 > 1$ as the case $m = 1$ holds by \cite[Proposition 1]{whiteley_et_al16}.} $\ug_{\rho-1} = g_{p,0} = g_{p}$, $\uK_{\rho} = \mathrm{Id}$, %K_{p,1}(\varphi) = \varphi$,
  $\uxi^{i}_{\rho-1} = \xi^{i}_{p,0}$ and $\uxi^{i}_{\rho} = \xi^{i}_{p,1}$. By Algorithm \ref{alg:AIRPF}, lines \ref{code line 3} and \ref{code line 2},
\begin{align}\label{eq:internal resampling with alpha}
    \uxi^{i}_{\rho} %= \xi^{i}_{n,1} &\sim \frac{\sum_{j=1}^N \alpha^{i,j}_0w^{j}_{n,0}g_n(\xi^{j}_{n,0})\delta_{\xi^{j}_{n,0}}}{\sum_{j=1}^N \alpha^{i,j}_0w^{i}_{n,0}g_n(\xi^{j}_{n,0})} \nonumber \\
    &\sim \frac{1}{w^{i}_{p,1}}\sum_{j=1}^N \alpha^{i,j}_0w^{j}_{p,0}\ug_{\rho-1}(\uxi^{j}_{\rho-1})\uK_{\rho}(\uxi^{j}_{\rho-1},\,\cdot\,).
\end{align}
One can check that $\ug_{\rho-1}\uK_{\rho}(\ol{\uQ}_{\rho,q}(\varphi)){w^{j}_{p,0}}/{w^{i}_{p,1}} = 
\ol{\uQ}_{\rho-1,q}(\varphi){\ol{\uw}^{j}_{\rho-1}}/{\ol{\uw}^{i}_{\rho}}$
and thus, by \eqref{eq:internal resampling with alpha},
 \begin{align} \label{eq:within process lack-of-bias}
    \E\left[\ol{\uQ}_{\rho,q}(\varphi)(\uxi^{i}_{\rho})\,\middle|\,\underline{\mathcal{G}}_{\rho-1}\right] %&= 
    &= \frac{1}{\ol{\uw}^{i}_{\rho}}\sum_{j=1}^{N}\alpha^{i,j}_0\ol{\uw}^{j}_{\rho-1} \ol{\uQ}_{\rho-1,q}(\varphi)(\uxi^{j}_{\rho-1})
\end{align}
and by \eqref{eq:2nd delta}, $\E[D_{\rho,q} \mid \underline{\mathcal{G}}_{\rho-1}] = 0$ for all $\rho \in \cR_{0}^q$.

Consider $\rho \in \cR_{+}^q$, and write $s = R_{S_m+1}(\rho-1)$. We have $\ug_{\rho-1} = 1$
for $\rho \in \cR_+^q$ and $\uK_{\rho} = \mathrm{Id}$ for $R_{S_m+1}(\rho-1) < S_m$ and $\uK_{\rho} = K_{I_{S_m+1}(\rho)}$ otherwise. Thus, similarly to above, by Algorithm \ref{alg:AIRPF}, lines \ref{code line 4} and \ref{code line 5},
\begin{align}
    \E\left[\ol{\uQ}_{\rho,q}(\varphi)(\uxi^{k,i}_{\rho})\,\middle|\,\underline{\mathcal{G}}_{\rho-1}\right] 
    &= \frac{1}{\ol{\uW}^{k}_{\rho}}\sum_{\ell=1}^{m} \alpha_{s}^{k,\ell}\ol{\uW}^{\ell}_{\rho-1}\ol{\uQ}_{\rho-1,q}(\varphi)(\uxi^{\ell,i}_{\rho-1}), \nonumber
\end{align}
and so $\E[D_{\rho,q} \mid \underline{\mathcal{G}}_{\rho-1}] = 0$ for all $\rho \in \cR_{+}^q$.

To prove \eqref{eq:martingale error bound}, we write the sum over $\rho \in \{0,\ldots,\un\}$ as a nested sum over $p \in \{0,\ldots,n\}$ and $s \in \{1,\ldots,S_m+1\}$ and apply Minkowski's inequality, yielding
\begin{multline*}
    \E\!\left[\left|\sum_{\rho=0}^{\un} D_{\rho,\un}\right|^r\right]^{\frac{1}{r}} 
    \leq \E\!\left[\left|\frac{1}{N}\sum_{i=1}^{N}\ol{\uw}_{0}^{i}\Delta^{i}_{0,\un}(\varphi)\right|^r\right]^{\frac{1}{r}} \\ + \sum_{p=0}^{n-1}\E\!\left[\left|\sum_{s=1}^{S_m+1} \frac{1}{N}\sum_{i=1}^{N}\ol{w}_{p,s}^{i}\Delta^{i}_{p,s,\un}(\varphi)\right|^r\right]^{\frac{1}{r}},
\end{multline*}
where
$\Delta^{j}_{p,s,q}(\varphi) = \Delta_{(S_m+1)p+s,q}^{I_{M}(j-1)+1,R_M(j-1)+1}(\varphi)$,
for all $j, s\in\N$, $q, p \in \N_{0}$. If we define for all $p\in \{0,\ldots,I_{S_m+1}(q)\}$, $M^{p,q}_{0}(\varphi)=0$, $\mathcal{G}^{p}_{0} = \underline{\mathcal{G}}_{p-1}$ and for all $\varrho+1 \in \{1,\ldots,S_mN\}$
\begin{align}\label{eq:martingale difference 2}
M^{p,q}_{\varrho+1}(\varphi) &= \ol{w}_{p,I_{N}(\varrho)+1}^{R_{N}(\varrho)+1}\Delta^{R_{N}(\varrho)+1}_{p,I_{N}(\varrho)+1,q}(\varphi),\quad\text{and}\quad
\mathcal{G}^{p}_{\varrho+1} = \mathcal{G}^{p}_{\varrho} \vee \sigma(\uxi^{R_{N}(\varrho)+1}_{p,I_{N}(\varrho)+1})
\end{align}
%Note that for any $\rho \in \{0,\ldots,\un\}$, $\ol{\uQ}_{\rho,\un} = \ol{\uQ}_{I_{S_m+1}(\rho),R_{S_m+1}(\rho),\un}$ and so, by the proof of  above, 
then we see that, because $\E[D_{\rho,\un} \mid \underline{\mathcal{G}}_{\rho-1}] = 0$, for all $\rho \in \{0,\ldots,\un\}$, the sequence $(M^{p,\un}_{\varrho}(\varphi),\mathcal{G}^{p}_{\varrho})_{\varrho \in \{0,\ldots,S_mN\}}$ is a martingale difference, for which the Burkholder-Davis-Gundy inequality \cite{burkholder_et_al72} yields
\begin{multline*}
    \E\left[\left|\sum_{s=1}^{S_m+1} \sum_{i=1}^{N}\frac{1}{N}\ol{w}_{p,s}^{i}\Delta^{i}_{p,s,\un}(\varphi)\right|^r \,\middle|\,\underline{\mathcal{G}}_{p-1}\right] \\ \leq 
    B^{\ast}_r\E\!\left[\!\left(\sum_{s=1}^{S_m+1} \sum_{i=1}^{N}\!\left(\frac{1}{N}\ol{w}_{p,s}^{i}\Delta^{i}_{p,s,\un}(\varphi)\right)^2\right)^{\frac{r}{2}}\,\middle|\,\underline{\mathcal{G}}_{p-1}\right]\!. %\label{eq:BDG application}
\end{multline*}
For some $B^{\ast}_{r} \in (0,+\infty)$. By $|\Delta^{i}_{p,s,\un}(\varphi)| \leq \osc\big(\overline{\uQ}_{p,s,\un}(\varphi)\big)$ and by \cite[Lemma 7.3.3]{delmoral04} as in  \cite{whiteley_et_al16},
\begin{align}
    \E\!\left[\left|\frac{1}{N}\sum_{i=1}^{N}\ol{\uw}_{0}^{i}\Delta^{i}_{0,\un}(\varphi)\right|^r\right]^{\frac{1}{r}} \leq \frac{d_{r}}{\sqrt{N}}\osc\big(\ol{\uQ}_{0,\un}(\varphi)\big),\label{eq:first term final}
\end{align}
because $\ol{w}^i_0 = 1$ for all $i\in \{1,\ldots,N\}$, from which the claim follows.
\end{proof}

Proposition \ref{prop:martingale construction} is analogous to \cite[Proposition 1]{whiteley_et_al16} and although the proof is similar to that of \cite[Proposition 1]{whiteley_et_al16}, the bound \eqref{eq:martingale error bound} is different as it involves a sum over the $S_m+1$ stages (including the internal resampling stage), which, as we shall see later in Theorem \ref{thm:stability}, translates into the $\sqrt{\log_2(m)/m}$ convergence rate of error, as $m\to +\infty$, which was originally reported in \cite{heine_et_al20,heine_et_al14}. Proposition 1 of \cite{whiteley_et_al16} uses a conditional version of Lemma 7.3.3 of \cite{delmoral04}, while we use the Burkholder-Davis-Gundy inequality. The conditional version of Lemma 7.3.3, together with Minkowski's inequality, would yield a bound  depending on $S_m+1$ with the the square root appearing \emph{inside} the sum over $s$, but with the Burkholder-Davis-Gundy inequality, the square root can be brought outside the sum, and thus \eqref{eq:martingale error bound} depends on $\sqrt{S_m+1}$, eventually leading to the convergence rate reported in \cite{heine_et_al20}. Thus, with the Burkholder-Davis-Gundy inequality, we obtain a tighter bound.

% -----------------------------------------------------
% STABILITY THEOREM
% -----------------------------------------------------
\begin{theorem}\label{thm:stability}
Fix $m,M \in \N$, and $N =mN$. Under Assumption \ref{ass:regularity}, there exist $C_1,~C_2(r) \in (0,+\infty)$, for any $r \in \{1,2,\ldots\}$, such that for all $\varphi \in \boundmeas(\X)$ %if, for some $\tau \in (0,1]$,
%\begin{align}\label{eq:ESS assumption}
%    \inf_{n\geq 0} \ESS^{N}_{n,0} \geq \tau,%\qquad \text{where}\qquad     \ESS^{N}_{n,s} = \frac{\left({N^{-1}}\sum_{i=1}^N w^{i}_{n,s}\right)^2}{{N^{-1}}\sum_{i=1}^N \left(w^{i}_{n,s}\right)^2},
%\end{align}
%then, 
\begin{align}%\label{eq:rel var theorem}
\inf_{n\geq 0} \ESS^{N}_{n,0} \geq \tau \implies 
\left\{
\arraycolsep=1.4pt
\begin{array}{rl}
    \sup_{n\geq 1} \E\left[\left({\uZ^{N}_n}/{\uZ_n}\right)^2\right]^{\frac{1}{n}} &\leq \displaystyle 1 + \frac{C_1}{mM\tau}, \\
    \sup_{n\geq 0} \E\left[\left|\pi^{N}_n(\varphi) - \pi_{n}(\varphi)\right|^{r}\right]^{\frac{1}{r}} &\leq \displaystyle C_2(r)\|\varphi\|\sqrt{\frac{\log_2(m)}{mM\tau}},
\end{array}\right.
\label{eq:uniform convergence}
\end{align}
where $\tau \in (0,1]$ and $Z^N_n = \gamma_n^{N}(1)$ and $Z_n = \gamma_n(1)$.
\end{theorem}
% -------------------------------------------------
\begin{proof}
For any $n\in \N$ and $0\leq \rho \leq \un$, where $\un = (S_m+1)n$ and $S_m = \log_2(m)$, we have $\uQ_{\rho,\un} = Q_{I_{S_m+1}(\rho-1)+1,n}$, and so, similarly to \cite[eq.~(39)]{whiteley_et_al16}, by Assumption \ref{ass:regularity},
\begin{align}
    \sup_{\substack{n\geq 0\\0 \leq \rho \leq n}}\osc\left(\ol{\uQ}_{\rho,n}(1)\right)^2 \leq 4\sup_{\substack{n\geq 0\\0 \leq \rho \leq n}}\sup_{x\in\X}\left |\ol{\uQ}_{\rho,n}(1)(x)\right|^2 < \delta\epsilon < +\infty,
\end{align}
By the proof of \cite[Theorem 2]{whiteley_et_al16}, Proposition~\ref{prop:bound for relative weight variance} below implies that by writing $C_1 = C\delta\epsilon$, $C$ being the constant in Proposition \ref{prop:bound for relative weight variance}, then
\begin{align}
    \E\left[\left({\underline{Z}^N_{\un}}/{\underline{Z}_{\un}}-1\right)^2\right] \leq \left(1+\frac{C_1}{mM\tau}\right)^{n} - 1,
\end{align}
which proves the first bound in \eqref{eq:uniform convergence}.

To prove the second bound in \eqref{eq:uniform convergence}, define
\begin{align*}
    \mu^{N}_p = \frac{1}{N}\sum_{i=1}^{N}w_{p,0}^{i}\delta_{\xi^{i}_{p,0}},\qquad \forall\,p\in \{0,1,\ldots,n\},
\end{align*}
and $\barphi = \varphi - \pi_n(\varphi)$, in which case $\pi^{N}_{n}(\varphi) - \pi_{n}(\varphi) = {\mu^{N}_n(\barphi)}/{\mu^{N}_n(1)}$.
%\	begin{align*}%\label{eq:final term of 2nd telescope}
%    \pi^{N}_{n}(\varphi) - \pi_{n}(\varphi) = \frac{\mu^{N}_n(\barphi)}{\mu^{N}_n(1)}.
%\end{align*}
By defining $M_p = \uQ_{p,0,\un}$ for all $p\in \{0,1,\ldots,n\}$ and by observing that $M_n = \mathrm{Id}$, we have
\begin{align}\label{eq:2nd telescope}
    \pi^{N}_{n}(\varphi) - \pi_{n}(\varphi) = \frac{\mu^{N}_0M_{0}(\barphi)}{\mu^{N}_0M_0(1)} + \sum_{p=1}^{n}\left(
    T^N_{p}(\barphi) - \frac{\mu^{N}_{p-1}M_{p-1}(\barphi)}{\mu^{N}_{p-1}M_{p-1}(1)}T_{p}^N(1)\right), %\left(\frac{\mu^{N}_pM_{p}(\barphi)}{\mu^{N}_p(1)} - \frac{\mu^{N}_{p-1}M_{p-1}(\barphi)}{\mu^{N}_{p-1}(1)}\right)
\end{align}
where
\begin{align*}
T^N_{p}(\varphi) = \frac{\mu^{N}_{p}M_{p}(\varphi)-\mu^{N}_{p-1}M_{p-1}(\varphi)}{\mu^{N}_{p}M_{p}(1)}, \qquad \forall~\varphi \in \boundmeas(\X).
\end{align*}
Similarly to the proof of~\cite[Theorem 2]{whiteley_et_al16}, and by following the proof of Proposition \ref{prop:martingale construction},
\begin{multline}\label{eq:minkowski form}
\E\left[\left|\pi^{N}_{n}(\varphi) - \pi_n(\varphi)\right|^r\right]^{\frac{1}{r}} \leq \frac{2d^{1/r}_r}{\sqrt{N}}C_{0,n}(\barphi) \\ + \sum_{p=1}^{n}4B^{\ast}_{r}C_{p,n}(\barphi)\E\left[\left(\sum_{s=1}^{S_m+1}\frac{1}{N\ESS^{N}_{p,s}}\right)^{\frac{r}{2}}\right]^{\frac{1}{r}},
\end{multline}
where
\begin{align}\label{eq:constant}
%\frac{\osc\left(Q_{p,n}(\barphi)\right)}{\inf_x Q_{p,n}(1)(x)} \leq 2C_{p,n}(\barphi),\quad\text{where}\quad 
C_{p,n} (\varphi) := \sup_{x,y}\frac{Q_{p,n}(1)(x)}{Q_{p,n}(1)(y)}\sup_{x}\left|\frac{Q_{p,n}(\varphi)(x)}{Q_{p,n}(1)(x)}\right|,\qquad \forall~\varphi \in \boundmeas(\X).
\end{align}
By Assumption \ref{ass:regularity} and because $\inf_{n\geq 0} \ESS^{N}_{n,0} \geq \tau$,
\begin{align}\label{eq:ESS stability}
    \mathcal{E}^N_{p,1} &= \frac{\left({N}^{-1}\sum_{i=1}^{N}g_p^i(\xi^{i}_{p,0})w^{i}_{p,0}\right)^2}{{N}^{-1}\sum_{i=1}^{N}\left(g_p^i(\xi^{i}_{p,0})w^{i}_{p,0}\right)^2} \nonumber \\ 
    &\geq \left(\frac{\inf_y g_p(y)}{\sup_{x} g_p(x)}\right)^2\frac{\left({N}^{-1}\sum_{i=1}^{N}\uw^{i}_{p,0}\right)^2}{{N}^{-1}\sum_{i=1}^{N}\left(\uw^{i}_{p,0}\right)^2} \geq \frac{\ESS^{N}_{p,0}}{\delta^2} \geq \frac{\tau}{\delta^2}.
\end{align}
One can check that $\ESS^{N}_{p,s+1} \geq \ESS^{N}_{p,s}$, and hence, $\ESS^{N}_{p,s} \geq {\tau}/{\delta^2}$ for all $s \in \{0,\ldots,S_m\}$. By plugging this into \eqref{eq:minkowski form}, we have 
\begin{align}
    \E\left[\left|\pi^{N}_{n}(\varphi) - \pi_n(\varphi)\right|^r\right]^{\frac{1}{r}} &\leq \frac{2d^{1/r}_r}{\sqrt{N}}C_{0,n}(\barphi) + 4B^{\ast}_{r}\sum_{p=1}^{n}C_{p,n}(\barphi)\sqrt{\frac{\delta^2\log_2(m)}{mM\tau}}. \nonumber
\end{align}
The proof is completed by using Assumption \ref{ass:regularity} similarly to \cite{whiteley_et_al16} (see also \cite{delmoral04}).
\end{proof}

% ---------------------------------------------------------------
% RELATIVE VARIANCE PROPOSITION
% ---------------------------------------------------------------
\begin{proposition}\label{prop:bound for relative weight variance}
Under Assumption \ref{ass:regularity}, there exists $C \in (0,+\infty)$, such that if, for any $m, M \in \N$, $p \in \N_0$, and $s \in \{0,\ldots,S_m\}$, where $S_m = \log_2(m)$, one has $\ESS^{mM}_{p,s} \geq \tau \in (0,1]$, then
\begin{align*}%\label{eq:variance proposition result}
    \E\!\left[\left({\uZ^N_n}/{\uZ_n} -  1\right)^2\right]  
    =\sum_{\rho=0}^{n-1} \frac{C}{mM\tau}\osc\big(\ol{\uQ}_{\rho,n}(1)\big)^2\left(\E\!\left[\left({\underline{Z}^{N}_{\rho}}/{\underline{Z}_{\rho}}-1\right)^2\right]+1\right), 
\end{align*}
where $\uZ_\rho = \ugamma_{\rho}(1)$ and $\uZ^{N}_{\rho} = N^{-1}\sum_{i=1}^{N}{\uw}^{i}_{\rho}$.
\end{proposition}
\begin{proof}
By setting $\varphi = 1$, and by the fact established in the proof of Proposition \ref{prop:martingale construction} that $(M^{p,n}_{\varrho}(1),\mathcal{G}^{p}_{\varrho})_{\varrho \in \{0,\ldots,S_mN\}}$ defined in \eqref{eq:martingale difference 2}, is a martingale difference, we have
\begin{multline}\label{eq:independent telescope}
    \E\!\left[\!\left(\frac{\uZ^N_n}{\uZ_n} -  1\right)^2\right] 
    =\sum_{\rho \in \cR_{0}^n \cup \{0\}}\E\!\left[\!\left(\frac{1}{N}\sum_{i=1}^{N}\ol{\uw}^{i}_{\rho}\Delta^{i}_{\rho,n}(1)\right)^2\right]  + \\\sum_{\rho\in \cR_+^n}\E\!\left[\!\left(\frac{1}{m}\sum_{k=1}^{m}\frac{1}{M}\sum_{i=1}^{M}\ol{\uw}_{\rho}^{k,i}\Delta^{k,i}_{\rho,n}(1)\right)^2\right].
\end{multline}
For $\rho \in \cR_{0}^n \cup \{0\}$, $\Delta_{\rho,n}^{1}(1),\ldots,\Delta_{\rho,n}^{N}(1)$ are conditionally independent, given $\underline{\mathcal{G}}_{\rho-1}$, and hence,
\begin{align*}
\E\!\left[\!\left(\frac{1}{N}\sum_{i=1}^{N}\ol{\uw}^{i}_{\rho}\Delta^{i}_{\rho,n}(1)\right)^2\right] = \sum_{i=1}^{N}\E\!\left[\!\left(\frac{1}{N}\ol{\uw}_{\rho}^{i}\Delta^{i}_{\rho,n}(1)\right)^2\right].
\end{align*}
By observing that $N^{-1}\sum_{i=1}^{N}\overline{\uw}^{i}_{\rho} = {\uZ^{N}_{\rho}}/{\uZ_{\rho}}$ and that by the proof of Proposition \ref{prop:martingale construction}, $\E[\uZ^{N}_\rho]=\uZ_\rho$ and by the assumption that $\ESS^{mM}_{p,s} \geq \tau$ for all $p \in \N_0$ and $s \in \{0,\ldots,S_m\}$, we have,  similarly to the proof of \cite[Proposition 2]{whiteley_et_al16}
\begin{align}
    \sum_{i=1}^{N}\E\!\left[\!\left(\frac{1}{N}\ol{\uw}_{\rho}^{i}\Delta^{i}_{\rho,n}(1)\right)^2\right] 
    &= \frac{\osc\big(\overline{\uQ}_{\rho,n}(1)\big)^2}{mM\tau}\E\left[\left(\frac{\uZ^{N}_{\rho}}{\uZ_\rho}-1\right)^2+1\right]. \label{eq:var proposition R0 part} 
\end{align}
For $\rho \in \cR_{+}^n$, and $k \in \{1,\ldots,m\}$, ${M}^{-1}\sum_{i=1}^{M}\ol{\uw}_{\rho}^{k,i}\Delta^{k,i}_{\rho,n}(1)$ are conditionally independent, given $\underline{\mathcal{G}}_{\rho-1}$, and hence, for the sum over $\cR_+^n$ in \eqref{eq:independent telescope} we have, by applying Lemma 7.3.3 of \cite{delmoral04}
\begin{align*}
&\E\!\left[\!\left(\frac{1}{m}\sum_{k=1}^{m}\frac{1}{M}\sum_{i=1}^{M}\ol{\uw}_{\rho}^{k,i}\Delta^{k,i}_{\rho,n}(1)\right)^2\right] \nonumber \\
&= \E\!\left[\!\frac{1}{m^2}\sum_{k=1}^{m}\E\!\left[\!\left(\frac{1}{M}\sum_{i=1}^{M}\ol{\uW}^k_{\rho}\Delta^{k,i}_{\rho,n}(1)\right)^2\,\middle|\, \underline{\mathcal{G}}_{\rho-1}\right]\right] \\
&\leq \frac{d_2}{M}\osc\!\left(\ol{\uQ}_{\rho,n}(1)\right)^2\E\!\left[\!\frac{1}{m^2}\sum_{k=1}^{m}\left(\ol{\uw}^{k,i}_\rho\right)^2\right],
\end{align*}
and, similarly to \eqref{eq:var proposition R0 part}, we have
\begin{align}
\E\!\left[\!\left(\frac{1}{m}\sum_{k=1}^{m}\frac{1}{M}\sum_{i=1}^{M}\ol{\uw}_{\rho}^{k,i}\Delta^{k,i}_{\rho,n}(1)\right)^2\right] &\leq 
    \frac{d_2\osc\big(\ol{\uQ}_{\rho,n}(1)\big)^2}{mM\tau}\E\!\left[\!\left(\frac{\uZ^{N}_{\rho}}{\uZ_\rho}-1\right)^2+1\right].\label{eq:var proposition R+ part}
\end{align}
The claim follows from \eqref{eq:independent telescope} -- \eqref{eq:var proposition R+ part} by noting that $\osc\big(\ol{\uQ}_{n,n}(1)\big)=0$.
\end{proof}

% ----------------------------------------------------
% ----------------------------------------------------
% INTERACTING VS INDEPENDENT PARTICLE FILTERS
% ----------------------------------------------------
% ----------------------------------------------------
\section{Interacting vs.~independent particle filters}
\label{sec:AIRPF vs IBPF}

We shall compare IBPF and AIRPF in three different ways: the ratio of the marginal likelihood estimate variances, effective number of filters, and the absolute variance of the marginal likelihood estimate. 

\subsection{Ratio of the marginal likelihood estimate variance}
\label{sec:variance comparison}

\newcommand{\IBPFest}{Z^{N,\mathrm{I}}_{n}}
\newcommand{\AIRPFest}{Z^{N,\mathrm{A}}_{n}}
\newcommand{\mAIRPF}{m^{\mathrm{A}}}
\newcommand{\mIBPF}{m^{\mathrm{I}}}

Let $\IBPFest$ and $\AIRPFest$ denote the marginal likelihood estimates of IBPF and AIRPF, respectively, and in addition, let $\mIBPF(n)$ and $\mAIRPF(n)$, respectively, denote the numbers of filters deployed by IBPF and AIRPF, as a function of the data record length $n$. For all filters, the sample size is $M$. 

If $\mAIRPF(n) \geq nC_1/\tau$, then by Theorem \ref{thm:stability} (see also \cite[Lemma 6, Remark 3]{whiteley_et_al16}, for details)
\begin{align}\label{eq:stability application}
  \inf_{p\in \N_0}\ESS^N_{p,0}>\tau \implies  \E\left[\left({\AIRPFest}/{Z_n}-1\right)^2\right] \leq \frac{2(S_m+1)}{M},
\end{align}
which implies that the ratio of the variances of $\IBPFest$ and $\AIRPFest$ satisfies
\begin{align*}
    \lim_{n\to+\infty} \frac{1}{n}\log \frac{\Var(\IBPFest)}{\Var(\AIRPFest)} 
    &\geq \lim_{n\to+\infty} \frac{1}{n}\log\frac{\displaystyle M\E\left[\left(Z^{M}_n/Z_n-1\right)^2\right]}{2(S_m+1)\mIBPF(n)} \nonumber \\ &\geq \Upsilon_{M}-\lim_{n\to+\infty}\frac{\log \mIBPF(n)}{n},%\label{eq:ratio limit} 
\end{align*}
where $Z^{M}_n$ is the marginal likelihood estimate of a single BPF with sample size $M$, and the first inequality follows by the independence of the $\mIBPF(n)$ filters of IBPF and \eqref{eq:stability application}. The second inequality follows from \cite[Proposition 4]{whiteley_lee14}, by which
\begin{align}\label{eq:twisted result}
    \lim_{n\to+\infty}\frac{1}{n}\log\E\left[\left({Z^{M}_n}/{Z_n}-1\right)^2\right] \geq \Upsilon_M > 0.
\end{align}
Thus, unless $\mIBPF(n)$ increases exponentially in $n$, the variance ratio will increase exponentially, while $\mAIRPF(n)$ increases only linearly. This conclusion is similar to that of \cite[Section 5]{whiteley_et_al16}, but we have demonstrated it to be true for AIRPF as well.

\subsection{Effective number of filters}
\label{sec:ESS}

Consider the ENF for $m$ independent BPFs with sample size $M$. By following the line of arguments for SIS in \cite{kong_et_al94}~(see also \cite[Proposition 1]{doucet_et_al00}), we can argue that 
\begin{align*}
\ESS^{mM}_{n,0} \approx (1 + \Var^\ast(Z^{M}_n))^{-1}\qquad\text{and}\qquad \Var^\ast(Z^{M}_n) \geq \Var^\ast(Z^{M}_{n-1}), \quad \forall~n \in \N,
\end{align*}
where $\Var^{\ast}$ denotes the variance over the probability space where the observations $Y_0,\ldots,Y_n$ random and not fixed as we have assumed so far. This suggests that the ENF for IBPF has a decreasing trend as $n\to \infty$, but we can make this notion more precise by observing that
\begin{align*}
    \lim_{n\to \infty}\lim_{m\to\infty}\ESS^{mM}_{n,0} 
    &= \lim_{n\to \infty}\frac{\left(\lim_{m\to\infty}{m}^{-1}\sum_{k=1}^{m}{W}^{k}_{n,0}\right)^2}{\lim_{m\to\infty}{m}^{-1}\sum_{k=1}^{m}\left({W}^{k}_{n,0}\right)^2} \nonumber \\
    &= \lim_{n\to \infty}\E\left[\left(\frac{Z^{M}_n}{Z_n}\right)^2\right]^{-1} = 0,
\end{align*}
where the final equation follows by \eqref{eq:twisted result}. Thus the asymptotic ENF (as $m \to\infty$) vanishes as $n\to\infty$, i.e.~for large enough $n$, only one filter will contribute to the IBPF output. For $M = 1$, ENF and ESS coincide and IBPF reduces to SIS. With AIRPF, such degeneracy can be avoided by forcing the ENF to remain above some pre-determined threshold.

\subsection{Absolute variance of the Marginal likelihood estimate}
\label{sec:collision}

\newcommand{\bk}{\boldsymbol{k}}
\newcommand{\bg}{\boldsymbol{g}}

For any $\varphi \in \boundmeas(\X)$, we define the Cartesian product function $\varphi^{\otimes 2}: (x,x') \mapsto \varphi(x)\varphi(x')$, and for any kernel $L: (x,\ud y) \mapsto L(x,\ud y)$ we define a product kernel $L^{\otimes 2}(x,x',\ud y \times \ud y') = L(x,\ud y)\otimes L(x',\ud y')$. A collision operator $\mathcal{C}_i$, is defined for $i \in \{0,1\}$ such that $\mathcal{C}_0(\varphi)(x,x') = \varphi(x,x')$ and $\mathcal{C}_1(\varphi)(x,x') = \varphi(x,x)$~\cite{cerou_et_al11}. For brevity, we shall write $\mathcal{C}_{(a_1,\ldots,a_p),(b_1,\ldots,b_p)} = \mathcal{C}_{\ones(a_1=b_1 \vee \cdots \vee a_p=b_p)}$ for any $p\in\N$. Moreover, define a kernel $\bR_n: \X^{2M} \times \mathcal{X}^{2M} \to [0,+\infty)$ as
\begin{align}\label{eq:resampling kernel}
&\bR_n(\varphi)(\bxi,\bxi') = \overline{\boldsymbol{g}}_{n}^{\otimes 2}(\bxi,\bxi') 
\sum_{\bu, \bv \in \{1,\ldots,M\}^M}  V_n^{\bu}(\bxi)V^{\bv}_n(\bxi') \varphi\!\left(\bxi^{\bu},{\bxi'}^{\bv}\right) 
\end{align}
for all $\bxi,\bxi' \in \X^{M}$ and $\varphi \in \boundmeas(\X^{2M})$, where $\bxi^{\bu} = (\xi^{u_1},\ldots,\xi^{u_M})$ and 
\begin{align*}
\ol{\bg}_n(\bxi) = \frac{1}{M}\sum_{i=1}^{M}g_n(\xi^{i})\qquad \text{and}\qquad V_n^{\bu}(\bxi) = \prod_{k=1}^{M} \frac{g_{n}(\xi^{u_k})}{\sum_{\ell=1}^{M}g_{n}(\xi^{\ell})} 
\end{align*}
for all $\bxi = (\xi^1,\ldots,\xi^M) \in \X^{M}$ and $\bu = (u_1,\ldots, u_M) \in \{1,\ldots,M\}^{M}$. To give some context, $\bR_n$ corresponds to the internal resampling within two filters that are chosen according to their weights. %For any $\varphi \in \boundmeas(\X)$  we shall also define 
%\begin{align*}
%\bvarphi(\bxi) = \sum_{i=1}^{M} \varphi(\xi^i), \qquad \forall~\bxi = (\xi^1,\ldots,\xi^M) \in \X^{\otimes M}.
%\end{align*}
For brevity, we also write $\sum_{i_1,\ldots,i_p} = \sum_{i_1=1}^{m}\cdots\sum_{i_p=1}^{m}$ for any $p \in \N$.

%
%
% Exact marginal likelihood variance formula
%
%
\begin{theorem}\label{thm:collision result}
For all $\varphi \in \boundmeas(\X)$, $n \in \N_0$, $m,M \in \N$ with $N = mM$ and $S_m = \log_2(m)$,
\begin{multline*}
    \E\!\left[\left(\sum_{i=1}^{N}W^{i}_{n,0}\varphi(\xi^{i}_{n,0})\right)^2\right] = \sum_{\bi,\bj}\!\left[\prod_{k=0}^{n-1} \prod_{s_k=1}^{S_m}  A^{i_{k,s_k+1}i_{k,s_k}}_{s_k}A^{j_{k,s_k+1}j_{k,s_k}}_{s_k}\right] \\ \times \pi^{\otimes 2M}_0\mathcal{C}_{i_{0,1},j_{0,1}}\bH_{0}^{\bi_{0},\bj_{0}} \cdots \bH_{n-1}^{\bi_{n-1},\bj_{n-1}}(\bvarphi^{\otimes 2}),
%     = \sum_{\bi,\bj}\!\left[\prod_{k=0}^{n-1} \Xi_k^{\bi_k,\bj_k} \right]\! \pi^{\otimes 2M}_0\mathcal{C}_{i_{0,1},j_{0,1}}\bH_{0}^{\bi_{0},\bj_{0}} \cdots \bH_{n-1}^{\bi_{n-1},\bj_{n-1}}(\bvarphi^{\otimes 2}) 
\end{multline*}
where $\bi = (\bi_{0},\ldots,\bi_{n-1})$, $\bj = (\bj_{0},\ldots,\bj_{n-1})$ and for all $0\leq q < n$, we define $\bi_q = (i_{q,1},\ldots,i_{q,S_m+1})$, $\bj_q = (j_{q,1},\ldots,j_{q,S_m+1})$, 
$$\bH^{\bi_q,\bj_q}_{q} = \bR_{q}\mathcal{C}_{(i_{q,1},\ldots,i_{q,S_m}),(j_{q,1},\ldots,j_{q,S_m})}\bK_{q+1}^{\otimes 2}\mathcal{C}_{i_{q,S_m+1},j_{q,S_m+1}},$$ and
\begin{align*}
\bvarphi(\bxi) = \sum_{i=1}^{M} \varphi(\xi^i), \qquad\forall~\bxi = (\xi^1,\ldots,\xi^M) \in \X^{ M}.
\end{align*}
\end{theorem}
\begin{proof}
The claim follows by writing
\begin{align*}
\E\left[\left(\sum_{i=1}^{N}W^{i}_{n,0}\varphi(\xi^{i}_{n,0})\right)^2\right] 
&= \sum_{i_{n,0},j_{n,0}}\E\left[W^{i_{n,0}}_{n,0}W^{j_{n,0}}_{n,0}\bvarphi(\bxi^{i_{n,0}}_{n,0})\bvarphi(\bxi^{j_{n,0}}_{n,0})\right], %\label{eq:MIPS collision 1 main}
\end{align*}
and repeatedly applying Lemma \ref{lem:collision induction 2}, first to the function $\bvarphi^{\otimes 2}$, then to the function $\bH_{n-1}^{\bi_{n-1},\bj_{n-1}}\bvarphi^{\otimes 2}$, etc., until finally applying Lemma \ref{lem:collision induction 2} to the function $\bH_{0}^{\bi_{0},\bj_{0}}\cdots\bH_{n-1}^{\bi_{n-1},\bj_{n-1}}\bvarphi^{\otimes 2}$.
\end{proof}

\begin{lemma}\label{lem:collision induction 2}
For all $n\in \N$, $\varphi \in \boundmeas(\X^{2M})$, $i_{S_m+1},j_{S_m+1} \in \{1,\ldots,m\}$ and $m,M \in \N$ with $N = mM$ and $S_m = \log_2(m)$,
\begin{multline*}
    \E\big[ W^{i_{S_m+1}}_{n+1,0}W^{j_{S_m+1}}_{n+1,0}{\varphi}(\bxi^{i_{S_m+1}}_{n+1,0},\bxi^{j_{S_m+1}}_{n+1,0})\big] \\ =\sum_{\bi,\bj}\left[\prod_{s=1}^{S_m}A^{i_{s+1}i_{s}}_{s}A^{j_{s+1}j_{s}}_{s}\right]\!\E\!\left[W^{i_{1}}_{n,0}W^{j_{1}}_{n,0}
    \bH_{n}^{\bi,\bj}\varphi\big(\bxi^{i_{1}}_{n,0},\bxi^{j_{1}}_{n,0}\big)\right]\!,
\end{multline*}
where $\bi = (i_1,\ldots,i_{S_m}),~\bj = (j_1,\ldots,j_{S_m}) \in \{1,\ldots,m\}^{S_m}$.
\end{lemma}
\begin{proof}
Let $\underline{\mathcal{G}}_0,\underline{\mathcal{G}}_1,\ldots$ be as defined in Proposition \ref{prop:martingale construction}. If $i_{S_m+1}\neq j_{S_m+1}$, then
\begin{multline}\label{eq:collision cases 1}
\E\!\left[\varphi(\bxi^{i_{S_m+1}}_{n,S_m+1},\bxi^{j_{S_m+1}}_{n,S_m+1}) \,\middle|\, \underline{\mathcal{G}}_{(S_m+1)n+S_m}\right] \\ = 
\sum_{{i_{S_m},j_{S_m}}}\frac{A^{i_{S_m+1},i_{S_m}}_{S_m}A^{j_{S_m+1},j_{S_m}}_{S_m}W^{i_{S_m}}_{n,S_m}W^{j_{S_m}}_{n,S_m}\bK_{n+1}^{\otimes 2}(\varphi)(\bxi^{i_{S_m}}_{n,S_m},\bxi^{j_{S_m}}_{n,S_m})}{W^{i_{S_m+1}}_{n,S_m+1}W^{j_{S_m+1}}_{n,S_m+1}},
\end{multline}
and if $i_{S_m+1}=j_{S_m+1}$, then 
\begin{multline}\label{eq:collision cases 2}
\E\!\left[\varphi(\bxi^{i_{S_m+1}}_{n,S_m+1},\bxi^{j_{S_m+1}}_{n,S_m+1}) \,\middle|\, \underline{\mathcal{G}}_{(S_m+1)n+S_m}\right] \\ = \sum_{i_{S_m}}\frac{A^{i_{S_m+1},i_{S_m}}_{S_m}W^{i_{S_m}}_{n,S_m}\bK_{n+1}\mathcal{C}_1(\varphi)(\bxi^{i_{S_m}}_{n,S_m})}{W^{i_{S_m+1}}_{n,S_m+1}}.
\end{multline}
In either case, 
\begin{align*}
&    \E\!\left[W^{i_{S_m+1}}_{n,S_m+1}W^{j_{S_m+1}}_{n,S_m+1}\varphi(\bxi^{i_{S_m+1}}_{n,S_m+1},\bxi^{j_{S_m+1}}_{n,S_m+1}) \,\middle|\, \underline{\mathcal{G}}_{(S_m+1)n+S_m}\right]  \\ &= \sum_{\substack{i_{S_m}\\j_{S_m}}}A^{i_{S_m+1},i_{S_m}}_{S_m}A^{j_{S_m+1},j_{S_m}}_{S_m}  W^{i_{S_m}}_{n,S_m}W^{j_{S_m}}_{n,S_m} \bK_{n+1}^{\otimes 2}\mathcal{C}_{i_{S_m+1},j_{S_m+1}}(\varphi)\big(\bxi^{i_{S_m}}_{n,S_m},\bxi^{j_{S_m}}_{n,S_m}\big).
\end{align*}
Thus, by Lemma \ref{lem:collision induction 1}
\begin{multline*}
\E\!\left[W^{i_{S_m+1}}_{n,S_m+1}W^{j_{S_m+1}}_{n,S_m+1}\varphi(\bxi^{i_{S_m+1}}_{n,S_m+1},\bxi^{j_{S_m+1}}_{n,S_m+1}) \,\middle|\, \underline{\mathcal{G}}_{(S_m+1)n+1}\right] \\
%\sum_{i_{S_m},j_{S_m}}\!\E\left[A^{i_{S_m+1},i_{S_m}}_{S_m}A^{j_{S_m+1},j_{S_m}}_{S_m}W^{i_{S_m}}_{n,S_m}W^{j_{S_m}}_{n,S_m}\bK_{n+1}^{\otimes 2}\mathcal{C}_{i_{S_m+1},j_{S_m+1}}\varphi\big(\bxi^{i_{S_m}}_{n,S_m},\bxi^{j_{S_m}}_{n,S_m}\big)\,\middle|\,\underline{\mathcal{G}}_{(S_m+1)(n+1)}\right] \\
= \sum_{\bi, \bj} \left[\prod_{q=1}^{S_m} A^{i_{q+1},i_{q}}_{q}A^{j_{q+1},j_{q}}_{q}\right]\!W^{i_1}_{n,1}W^{j_1}_{n,1}\mathcal{C}_{\bi',\bj'}\bK_{n+1}^{\otimes 2}\mathcal{C}_{i_{S_m+1},j_{S_m+1}}(\varphi)(\bxi^{i_1}_{n,1},\bxi^{j_1}_{n,1}),
% \label{eq:plug it here}
\end{multline*}
where $\bi' = (i_2,\ldots,i_{S_m})$ and $\bj' = (j_2,\ldots,j_{S_m})$. Finally,
\begin{multline*}
    W_{n,1}^{i_1}W_{n,1}^{j_1}\E\!\left[\mathcal{C}_{\bi_{2:S_m},\bj_{2:S_m}}\bK_{n+1}^{\otimes 2}\mathcal{C}_{i_{S_m+1},j_{S_m+1}}\varphi(\bxi^{i_1}_{n,1},\bxi^{j_1}_{n,1})\,\big|\,\underline{\mathcal{G}}_{(S_m+1)n}\right] \\
    =
    W^{i_{1}}_{n,0}W^{j_{1}}_{n,0}\bH_{n}^{\bi_{n},\bj_{n}}(\varphi)\!\left(\bxi^{i_{1}}_{n,0},\bxi^{j_{1}}_{n,0}\right),
%    \label{eq:introducing UV sum}
\end{multline*}
since $W^{i_{1}}_{n,1}W^{j_{1}}_{n,1} = W^{i_{1}}_{n,0}W^{j_{1}}_{n,0}\overline{\bg}_{n}^{\otimes 2}(\bxi^{i_{1}}_{n,0},\bxi^{j_{1}}_{n,0})$. The claim then follows.
%\todo{Note: $i_0 = i_1$ and so is $j_0 = j_1$. Therefore we do not have a sum over $i_0$ or $j_0$.}
\end{proof}

\begin{lemma}\label{lem:collision induction 1}
For all $\varphi \in \boundmeas(\X^{2M})$, $n \in \N$, $s\in\{1,\ldots,S_m-1\}$ and any $i_{S_m+1}, j_{S_m+1} \in \{1,\ldots,m\}$,
\begin{multline*}
    \sum_{i_{S_m},j_{S_m}}\E\left[A^{i_{S_m+1},i_{S_m}}_{S_m}A^{j_{S_m+1},j_{S_m}}_{S_m}W^{i_{S_m}}_{n,S_m}W^{j_{S_m}}_{n,S_m}\varphi(\bxi^{i_{S_m}}_{n,S_m},\bxi^{j_{S_m}}_{n,S_m})\,\middle|\,\underline{\mathcal{G}}_{(S_m+1)n+ s} \right] \\
    = \sum_{\bi_{s:S_m},\bj_{s:S_m}} \left[\prod_{q=s}^{S_m} A^{i_{q+1}i_{q}}_{q}A^{j_{q+1}j_{q}}_{q}\right]\!W^{i_s}_{n,s}W^{j_s}_{n,s}\mathcal{C}_{\boldsymbol{i}_{s+1:S_m},\boldsymbol{j}_{s+1:S_m}}(\varphi)(\bxi^{i_s}_{n,s},\bxi^{j_s}_{n,s}),
\end{multline*}
where $\bi_{s+1:S_m} = (i_{s+1},\ldots,i_{S_m})$ and $\bj_{s+1:S_m} = (j_{s+1},\ldots,j_{S_m})$.
\end{lemma}
\begin{proof}
The proof is by backward induction. The induction start at rank $s = S_m-1$ follows immediately by \eqref{eq:collision cases 1} -- \eqref{eq:collision cases 2} by replacing $\bK^{\otimes 2}_{n+1}$ with the identity kernel. Assume the claim to hold at rank $s$. In this case,
\begin{align}
 &   \sum_{i_{S_m},j_{S_m}}\E\left[A^{i_{S_m+1},i_{S_m}}_{S_m}A^{j_{S_m+1},j_{S_m}}_{S_m}W^{i_{S_m}}_{n,S_m}W^{j_{S_m}}_{n,S_m}\varphi(\bxi^{i_{S_m}}_{n,S_m},\bxi^{j_{S_m}}_{n,S_m})\,\middle|\,\mathcal{G}_{(S_m+1)n + s-1} \right]  \nonumber\\
  &  = \sum_{\substack{\bi_{s:S_m}\\\bj_{s:S_m}}} \left[\prod_{q=s}^{S_m} A^{i_{q+1}i_{q}}_{q}A^{j_{q+1}j_{q}}_{q}\right] W^{i_s}_{n,s}W^{j_s}_{n,s}  \E \left[\mathcal{C}_{\bi'_{s},\bj'_{s}}(\varphi)(\bxi^{i_s}_{s},\bxi^{j_s}_{s})\,\middle|\,\mathcal{G}_{(S_m+1)n + s-1} \right],   \label{eq:induction step target}
\end{align}
where we have used the shorthand notations $\bi'_s = \bi_{s+1:S_m}$ and $\bj'_s = \bj_{s+1:S_m}$. Similarly to the proof of Lemma \ref{lem:collision induction 2},
\begin{align}
    &W^{i_s}_{n,s}W^{j_s}_{n,s}\E\left[\mathcal{C}_{\bi_{s+1:S_m},\bj_{s+1:S_m}}(\varphi)(\bxi^{i_s}_{n,s},\bxi^{j_s}_{n,s})\,\middle|\,\mathcal{G}_{(S_m+1)n + s-1} \right] \nonumber \\
    &= \sum_{i_{s-1},j_{s-1}} A^{i_{s}i_{s-1}}_{s-1}A^{j_{s}j_{s-1}}_{s-1}W^{i_{s-1}}_{n,s-1}W^{j_{s-1}}_{n,s-1}\mathcal{C}_{\bi_{s:S_m},\bj_{s:S_m}}(\varphi)(\bxi^{i_{s-1}}_{n,s-1},\bxi^{j_{s-1}}_{n,s-1}),\label{eq:coll 2nd substitution}
\end{align}
and we have the claim by substituting \eqref{eq:coll 2nd substitution} into \eqref{eq:induction step target}.
\end{proof}

Theorem \ref{thm:collision result} is reminiscent to Lemma 4 of \cite{heine_et_al17}, which can be directly applied to calculating the variance for IBPF and ARPF, where individual particles at any stage are conditionally independent, given the previous stage. In AIRPF, this is not true because resampling is done across subsets (islands) of $M$ particles which makes the particles within the subsets dependent. The conditional independence holds in AIRPF, but only  at the filter level, and therefore Theorem \ref{thm:collision result} involves collisions between filters rather than collisions between particles. Due to the internal resampling, particle level interactions too need to be taken into account and this is accomplished by the kernel $\bR_{n}$, defined in \eqref{eq:resampling kernel}. This combined particle and filter level collision analysis is the main difference to Lemma 4 of \cite{heine_et_al17}.

\section{Numerical Experiments}
\label{sec:numerical experiments}

\subsection{Exact variance}

To demonstrate the results of Section \ref{sec:collision}, consider a simple HMM on $\X = \{0,1\}$:
\begin{align*}
    X_0 \sim \frac{1}{2}\delta_{0} + \frac{1}{2}\delta_{1},\quad X_{n+1} \mid X_{n} &= x_n \sim \frac{3}{4}\delta_{x_n}+\frac{1}{4}\delta_{1-x_n}, \\
    Y_n \mid X_n &= x_n \sim \frac{3}{4}\delta_{x_n} + \frac{1}{4}\delta_{1-x_n}.%\end{array}
\end{align*}
for all  $n\in \N_0$. This model is not of any practical importance, but allows us to demonstrate the performance of the various estimators in a non-approximate manner. We shall  consider the classical BPF and three different parallel particle filter algorithms: (1) IBPF, (2) augmented resampling particle filter (ARPF) of~\cite{heine_et_al20}, and (3) AIRPF. We do not regard ARPF as a practically feasible alternative for AIRPF, but it has been included to demonstrate the impact of the level of interaction to the variance of the marginal likelihood estimator.

As our purpose with this example is only to demonstrate the impact of the different interaction schemes to the variance of the marginal likelihood estimator, we will keep the scenario simple by choosing $m=4$ and $M=2$ making the total sample size $N=8$. 

Figure \ref{fig:toy example} shows the marginal likelihood estimator variances for IBPF, ARPF and AIRPF relative to that of the classical BPF with $N=8$. All relative variances are above 1, as BPF allows the full interaction between particles while the IBPF, ARPF and AIRPF constrain the interactions and therefore BPF yields the lowest variance. IBPF, without any interaction between the $m$ filters, has the highest variance with seemingly exponential growth as predicted in Section \ref{sec:variance comparison}. The performances of AIRPF and ARPF are intermediate, and suggest sub-exponential growth.

\begin{figure}[t]
\begin{center}
\includegraphics{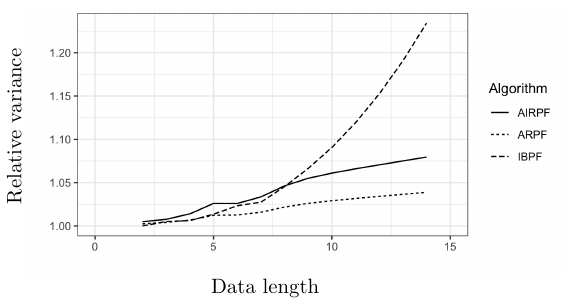}
%\begin{tikzpicture}
%\node at (0,0) {\includegraphics[width=0.75\textwidth]{toy_model_results.png}};
%\node at (-0.5,-3) {Data length};
%\node[rotate=90] at (-6.5,0.25) {Relative variance};
%\end{tikzpicture}
\end{center}
\caption{Variances of the marginal likelihood estimators, relative to BPF.}
\label{fig:toy example}
\end{figure}

Note that computational time considerations are intentionally omitted, as our purpose is to demonstrate different orders of complexity between IBPF, ARPF, and AIRPF, which suggests that asymptotically AIRPF will be superior to IBPF for longer data records. 

\subsection{PMCMC Application}
\label{sec:change point}

In our second example, we consider simultaneous change point detection and the estimation of the probability of change point occurring in sequential counting data. The data represents the number of on-line news articles that contain a specific keyword. The data contains daily counts of articles containing the keyword and the total number of articles published on the day. The model for the daily rate of articles containing the keyword is a $[0,1]$-valued Markov chain $X_0,X_1,\ldots$ with the initial distribution $X_0 \sim f_{\alpha,\beta}$, where $f_{\alpha,\beta}$ is the law of $\beta$-distribution with parameters $\alpha, \beta \in \R$. For the signal kernel, we set
\begin{align*}
    K(x,\ud x') = (1-p)\delta_{x}(\ud x') + pf_{\alpha,\beta}(\ud x'), \qquad \forall~x\in[0,1],
\end{align*}
We assume $\alpha$ and $\beta$ to be known, but $p$ to be unknown. We impose another $\beta$-prior for $p$ with known parameters. %\todo{Double check the initial distribution}

We shall estimate the joint Bayesian posterior law of $p$ and the occurrences of the change points PMMH, deploying either IBPF or AIRPF. The experiment was carried out for $m = 64$ filters with two different sample sizes, $M=200$ and $M=500$. 

\begin{figure}[t!]
\begin{center}
\includegraphics[width=0.95\textwidth]{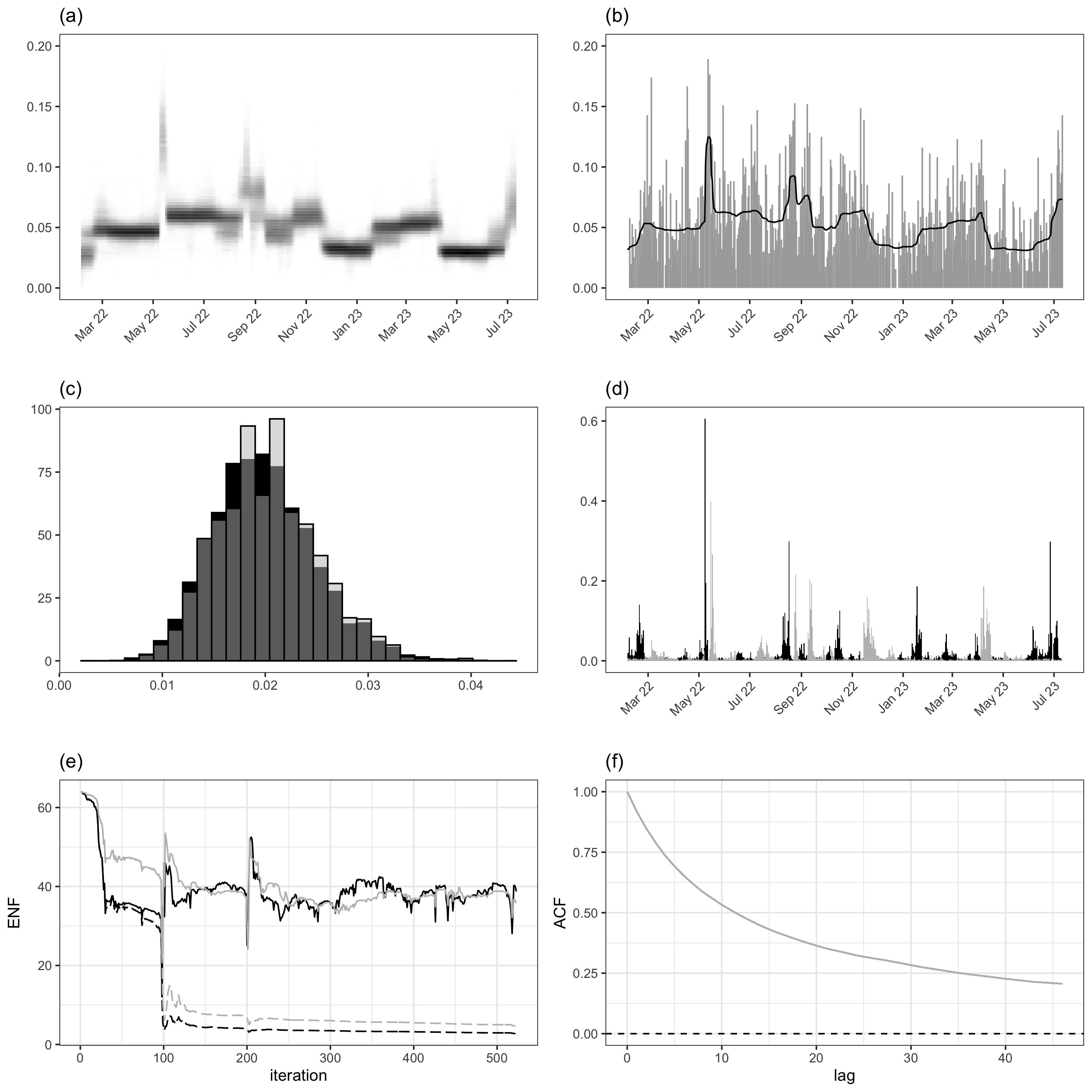}
\end{center}
\caption{(a) Marginal posterior distribution of $X_t$. (b) The data and the marginal mean of $X_t$. (c) Posterior distribution of $p$ with IBPF (transparent gray) and AIRPF (black). (d) Marginal probability of change points (black: upwards, gray: downwards). (e) Average ENF for IBPF (dahsed) and AIRPF (solid) across the 523 filter iterations. (f) ACF of $p$ for IBPF (dashed) and AIRPF (solid). In (e) and (f), black depicts $M=200$ and gray $M=500$.}
\label{fig:change points}
\end{figure}

Figure \ref{fig:change points} shows the results for 50000 iterations of PMMH. Except for the approximate posterior for $p$ in Figure \ref{fig:change points}(c), the results are shown for AIRPF only. The omitted results for IBPF are broadly similar. Figure \ref{fig:change points} demonstrates  AIRPF to be suitable for PMMH, but we refrain from claiming PMMH in general to be the preferred method for the change point detection, as many other methods exist, such as reversible jump Markov chain Monte Carlo~\cite{green95}. 

To compare IBPF with AIRPF, the ENF and the estimated autocorrelation function of the resulting Markov chain are shown in Figures \ref{fig:change points}(e) -- (f). The AIRPF deploys adaptive resampling, i.e.~the interaction between different filters was executed only when the ENF went below a predefined threshold value of $0.3\times m$. AIRPF maintains the ENF at the required level while the ENF of IBPF decreases as predicted.

We have included the results for two different sample sizes $M=200$ and $M=500$ to demonstrate the effect of $M$ on the comparative performance. For large $M$, IBPF and AIRPF produce similar results, as expected; for large $M$, running $m$ filters adds little to the accuracy of the marginal likelihood estimator as a single filter with large $M$ already gives a sufficiently accurate estimate, but for small $M$, multiple filters incur a more notable improvement.

%%%%%%%%%%%%%%%%%%%%%%%%%%%%%%%%%%%%%%%%%%%%%%%
%%% Example with single Appendix:            %%
%%%%%%%%%%%%%%%%%%%%%%%%%%%%%%%%%%%%%%%%%%%%%%%
%\begin{appendix}
%\section*{Title}\label{appn} %% if no title is needed, leave empty \section*{}.
%Appendices should be provided in \verb|{appendix}| environment,
%before Acknowledgements.
%
%If there is only one appendix,
%then please refer to it in text as \ldots\ in the \hyperref[appn]{Appendix}.
%\end{appendix}
%%%%%%%%%%%%%%%%%%%%%%%%%%%%%%%%%%%%%%%%%%%%%%%
%%% Example with multiple Appendixes:        %%
%%%%%%%%%%%%%%%%%%%%%%%%%%%%%%%%%%%%%%%%%%%%%%%
%\begin{appendix}
%\section{Title of the first appendix}\label{appA}
%If there are more than one appendix, then please refer to it
%as \ldots\ in Appendix \ref{appA}, Appendix \ref{appB}, etc.
%
%\section{Title of the second appendix}\label{appB}
%\subsection{First subsection of Appendix \protect\ref{appB}}
%
%Use the standard \LaTeX\ commands for headings in \verb|{appendix}|.
%Headings and other objects will be numbered automatically.
%\begin{equation}
%\mathcal{P}=(j_{k,1},j_{k,2},\dots,j_{k,m(k)}). \label{path}
%\end{equation}
%
%Sample of cross-reference to the formula (\ref{path}) in Appendix \ref{appB}.
%\end{appendix}

%%%%%%%%%%%%%%%%%%%%%%%%%%%%%%%%%%%%%%%%%%%%%%
%% Acknowledgements                         %%
%% should be provided in the                %%
%% Acknowledgements section.                %%
%%%%%%%%%%%%%%%%%%%%%%%%%%%%%%%%%%%%%%%%%%%%%%
\begin{acks}[Acknowledgments]
This research made use of the Balena High Performance Computing Service at the Universityof Bath.
\end{acks}

%%%%%%%%%%%%%%%%%%%%%%%%%%%%%%%%%%%%%%%%%%%%%%
%% Funding information, if any,             %%
%% should be provided in the                %%
%% funding section.                         %%
%%%%%%%%%%%%%%%%%%%%%%%%%%%%%%%%%%%%%%%%%%%%%%
%\begin{funding}
%The first author was supported by NSF Grant DMS-??-??????.
%
%The second author was supported in part by NIH Grant ???????????.
%\end{funding}

%%%%%%%%%%%%%%%%%%%%%%%%%%%%%%%%%%%%%%%%%%%%%%
%% Supplementary Material, including data   %%
%% sets and code, should be provided in     %%
%% {supplement} environment with title      %%
%% and short description. It cannot be      %%
%% available exclusively as external link.  %%
%% All Supplementary Material must be       %%
%% available to the reader on Project       %%
%% Euclid with the published article.       %%
%%%%%%%%%%%%%%%%%%%%%%%%%%%%%%%%%%%%%%%%%%%%%%
\begin{supplement}
\stitle{Codes and data}
\sdescription{The data set and the source code for Section \ref{sec:change point}, are available at \url{https://github.com/heinekmp/AIRPF_for_PMCMC}}
\end{supplement}
%\begin{supplement}
%\stitle{Title of Supplement B}
%\sdescription{Short description of Supplement B.}
%\end{supplement}

%%  MR numbers will be added by VTeX.                      %%
%%                                                         %%
%%  Use \cite{...} to cite references in text.             %%
%%                                                         %%
%%%%%%%%%%%%%%%%%%%%%%%%%%%%%%%%%%%%%%%%%%%%%%%%%%%%%%%%%%%%%

%% if your bibliography is in bibtex format, uncomment commands:
\bibliographystyle{imsart-number} % Style BST file (imsart-number.bst or imsart-nameyear.bst)
\bibliography{my.bib}       % Bibliography file (usually '*.bib')

%%% or include bibliography directly:
%\begin{thebibliography}{9}
%
%\bibitem{r1}
%\textsc{Billingsley, P.} (1999). \textit{Convergence of
%Probability Measures}, 2nd ed.
%Wiley, New York.
%\MR{1700749}
%
%\bibitem{r2}
%\textsc{Bourbaki, N.}  (1966). \textit{General Topology}  \textbf{1}.
%Addison--Wesley, Reading, MA.
%
%\bibitem{r3}
%\textsc{Ethier, S. N.} and \textsc{Kurtz, T. G.} (1985).
%\textit{Markov Processes: Characterization and Convergence}.
%Wiley, New York.
%\MR{838085}
%
%\bibitem{r4}
%\textsc{Prokhorov, Yu.} (1956).
%Convergence of random processes and limit theorems in probability
%theory. \textit{Theory  Probab.  Appl.}
%\textbf{1} 157--214.
%\MR{84896}
%\end{thebibliography}

\end{document}